\documentclass [letterpaper] {article}
\pdfoutput=1
\usepackage{amsfonts}
\usepackage{amssymb}
\usepackage{amsmath}
\usepackage{amsthm}
\usepackage{url}
\usepackage[pdftex]{color, graphicx}
\usepackage{subfigure}
\usepackage{algorithm}
\usepackage{algorithmic}
\usepackage [letterpaper] {geometry}
\usepackage[nottoc, notlot, notlof]{tocbibind}
\usepackage{indentfirst}

\newtheorem{theorem}{Theorem}[section]
\DeclareMathOperator{\dist}{dist}
\DeclareMathOperator{\grad}{\nabla}
\DeclareMathOperator{\sign}{sign}

\newcommand{\set}[1]{\{\,#1\,\}}
\newcommand{\R}{\mathbb{R}}

\title{A Branch and Cut Algorithm for the Halfspace Depth Problem\footnote{This research was partly supported by NSERC Canada. Computational Resources are supplied by ACEnet.}}
\author{David Bremner \and Dan Chen}
\date{}
\begin{document}
%\large
\maketitle

\begin{abstract}
The concept of \emph{data depth} in non-parametric multivariate descriptive statistics is the generalization of the univariate rank method to multivariate data. \emph{Halfspace depth} is a measure of data depth. Given a set $S$ of points and a point $p$, the halfspace depth (or rank) of $p$ is defined as the minimum number of points of $S$ contained in any closed halfspace with $p$ on its boundary. Computing halfspace depth is NP-hard, and it is equivalent to the Maximum Feasible Subsystem problem. In this paper a mixed integer program is formulated with the big-$M$ method for the halfspace depth problem. We suggest a branch and cut algorithm for these integer programs. In this algorithm, Chinneck's heuristic algorithm is used to find an upper bound and a related technique based on sensitivity analysis is used for branching. Irreducible Infeasible Subsystem (IIS) hitting set cuts are applied. We also suggest a binary search algorithm which may be more numerically stable. The algorithms are implemented with the BCP framework from the \textbf{COIN-OR} project.
\end{abstract}

\section{Introduction}
\label{sec:intro}
\emph{Halfspace depth} is a measure of \emph{data depth} which is a term from non-parametric multivariate descriptive statistics. The depth or rank of a point (a tuple data item) measures the centrality of this point with respect to a given set of points in a high dimensional space. The data with the largest depth is considered the median of the data set, which best describes the data set. Since tuple data items can be represented as points in Euclidean space $\mathbb{R}^{d}$, these two terms are used interchangeably in this paper.

Halfspace depth is also called Tukey depth or location depth. Given a set $S$ of points and a point $p$ in $\mathbb{R}^{d}$, the halfspace depth of $p$ is defined as the minimum number of points of $S$ contained in any closed halfspace with $p$ on its boundary. The point with the largest depth is called \emph{halfspace median} or \emph{Tukey median}.
\begin{figure}[!ht]
  \centering
  \includegraphics[width=0.35\textwidth]{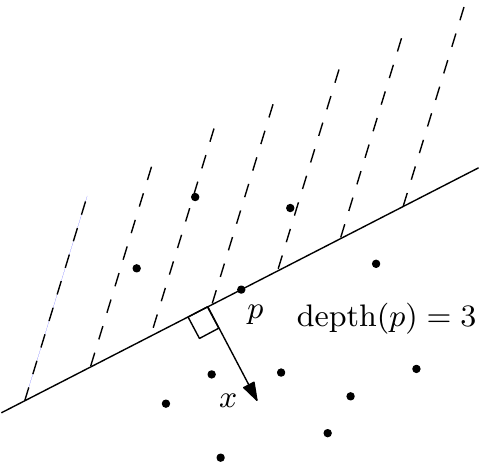}
  \caption{An example of halfspace depth in $\mathbb{R}^{2}$}
  \label{fig:hd1}
\end{figure}
In Figure~\ref{fig:hd1}, the halfspace depth of $p$ is $3$, because at least three points will be contained by any closed halfspace with $p$ on its boundary. The halfspace depth of $p$ can also be written as:
\begin{equation}
  \label{eq:1.1}
  \min_{x \in \mathbb{R}^{d} \backslash 0} | \{ q \in S  | x \cdot q \leq x \cdot p \} |
\end{equation}
where $x$ is the outward normal vector of the closed halfspace. This optimization problem is equivalent to  
\begin{equation}
  \label{eq:1.2}
  | S | - \max_{x \in \mathbb{R}^{d}} | \{ q \in S  | x \cdot q > x \cdot p \} |
\end{equation}
When a point is excluded from the halfspace, the corresponding inequality in \eqref{eq:1.2} is satisfied. The problem is thus to find a vector $x$ that maximizes the number of satisfied inequalities.

A data set is said to be in general position if it has no ties, i.e. no more than two points on the same line, no more than three points on the same plane, and so forth. If the data set is in general position, the halfspace depth problem is identical to the \emph{open hemisphere problem} introduced by Johnson and Preparata. Given a set of $n$ points on the unit sphere $S^{d}$ in $\mathbb{R}^{d}$, the open hemisphere problem is to find an open hemisphere of $S^{d}$ that contains as many points as possible. This problem is NP-complete if both $n$ and $d$ are parts of the input~\cite{Johnson}.

\subsection{Organization of This Paper}
\label{sec:intro.over}

This paper is organized as follows. In Section~\ref{sec:anova} we discuss on application of halfspace depth in non-parametric statistics. In Section~\ref{sec:mfs} we show that the halfspace depth problem is equivalent to the \emph{maximum feasible subsystem (MAX FS)} problem. In Section~\ref{sec:mip} we discuss different integer program formulations for the halfspace depth problem. In Section~\ref{sec:heur} we introduce Chinneck's heuristic algorithm for the MAX FS problem. In Section~\ref{sec:alg} we introduce our branch and cut algorithm. Unfortunately, in some circumstances this algorithm encounters numerical difficulties. For a more accurate result, we also introduce a binary search strategy in this section. In Section~\ref{sec:impl} we describe the details of our implementation using BCP; the BCP framework is also briefly introduced in this section. In Section~\ref{sec:test} we give some testing results and benchmark the performance of our algorithm. In Section~\ref{sec:fuwok} we list some future work for this project.

\section{Application: Two Factor ANOVA}
\label{sec:anova}
In this section we describe an application of halfspace depth to ANOVA (ANalysis Of VAriance) testing.  For more on the use of halfspace depth in this context, see~\cite{mizera2002,mizera:anova}. We consider here \emph{two factor} ANOVA testing, where an experiment has two different \emph{experimental factors} with \emph{levels} in $N=\set{ 1 \dots n }$ and $M= \set{1 \dots m}$. For example we may have $n$ soil types and $m$ fertilizers, and we wish to see how various combinations affect crop yield. For each experimental setting $(i,j)$ we have $r$ data points $z_{i,j,1} \dots z_{i,j,r}$ measuring outcomes. The subset $\set{ z_{i,j,k} \mid k=1\dots r}$ corresponding to an experimental scenario is fit to (i.e.\ approximated by) a linear function $\mu_i+\nu_j$. Our space of possible fits can be be identified with $\R^{n+m}$ via the parameter vector $\vartheta=(\mu_1 \dots \mu_n, \nu_1 \dots \nu_m)$. The quality of a given fit is evaluated using \emph{criterial functions}, typically using squared Euclidean distance like the following:
\begin{equation*}
  F_{i,j,k}(\vartheta)=\frac{(z_{i,j,k}-(\mu_i+\nu_j))^2}{2} \, .
\end{equation*}
It turns out that halfspace depth can be used to evaluate the local optimality of a given fit $\vartheta$. If there is another fit $\vartheta'$ in the neighbourhood of $\vartheta$ that improves (i.e. decreases) every criterial function, then $\vartheta$ is clearly not optimal. To quantify this observation, we want to measure, over all possible directions of perturbation for $\vartheta$, the maximum number of criterial functions that decrease. In terms of the gradients (i.e.\ vectors of partial derivatives) $\grad F_{i,j,k}(\vartheta)$ of our criterial functions, the change in $F_{i,j,k}$ as we perturb in direction $v$ is given by the inner product $\grad F_{i,j,k}(\vartheta)\cdot v$. Since we want to decrease criterial functions, we are looking for a direction $v$ such that the corresponding linear halfspace $v^T x \geq 0$ contains as few of the $\grad F_{i,j,k}(\vartheta)$ as possible. This is equivalent to asking for the halfspace depth of the origin in the set  of gradients (for fixed fit $\vartheta$, the gradients will be constant vectors in $\R^{n+m}$). In our case the gradients have a simple form 
\begin{align*}
  \grad F_{i,j,k}(\vartheta)&=-(z_{i,j,k}-\mu_i -\mu_j)(e_i+e_j)\\
\intertext{ Since we only care about the sign, we may scale the gradients arbitrarily and consider instead the $\set{0,\pm 1}$ vectors}
  G_{i,j,k}(\vartheta)&=\sign \grad F_{i,j,k} \\
  &=-\sign(z_{i,j,k}-\mu_i -\mu_j) (e_i+e_j)  
\end{align*}
We give computational results for some randomly generated examples of this type in Section~\ref{sec:test.2}.

\section{Maximum Feasible Subsystem}
\label{sec:mfs}
The halfspace depth problem has a strong connection with the Maximum Feasible Subsystem (MAX FS) problem. Given an infeasible linear system, the MAX FS problem is to find a maximum cardinality feasible subsystem. This problem is NP-hard~\cite{Chakravarti,Sankaran}, and also hard to approximate~\cite{Amaldi}.

When $p$ is contained in the convex hull of $S$, and $p$ is on the boundary of a closed halfspace, as illustrated in Figure~\ref{fig:hd1}, there must be some data contained by the halfspace. Then the set of inequalities
\begin{equation}
  \label{eq:2.1}
   x \cdot q >  x \cdot p \qquad \forall q \in S
\end{equation}
or
\begin{equation}
  \label{eq:2.2}
  x \cdot ( q - p ) > 0 \qquad \forall q \in S
\end{equation}
in \eqref{eq:1.2} can not all be satisfied at the same time, i.e.~\eqref{eq:2.2} is an infeasible linear system. Now it is clear that to compute the halfspace depth of $p$ is to find the maximum number of inequalities in \eqref{eq:2.2} that can be satisfied at the same time; in other words to find the maximum feasible subsystem of \eqref{eq:2.2}. Therefore, the halfspace depth problem is a MAX FS problem. The MAX FS problem can also be seen as finding a minimum cardinality set of constraints, whose removal makes the original infeasible system feasible. This problem is called the \emph{minimum unsatisfied linear relation (MIN ULR)} problem~\cite{Chinneck}.

\subsection{Irreducible Infeasible Subsystems}
\label{sec:mfs.iis}
In an infeasible linear system, an \emph{irreducible infeasible subsystem (IIS)} is a subset of constraints that itself is infeasible, but any proper subsystem is feasible. If a subset of points $A$ of $S$ forms a minimal simplex which contains $p$, the inequalities in \eqref{eq:2.2} defined by $A$ form an IIS. The point set $A$ is a \emph{minimal dominating set (MDS)}, which is is a minimal set of points whose convex hull contains $p$~\cite{David}.

Every infeasible system contains one or more IISs. To make the original system feasible, we need to delete at least one inequality from every IIS, in other words, we need to delete a hitting set of all IISs in the infeasible system. The \emph{minimum-cardinality IIS   set-covering (MIN IIS COVER)} problem is to find the smallest cardinality set of constraints to hit all IISs of the original system (this problem is a \emph{minimum hitting set} problem in the terminology of e.g.~\cite{GareyJohnson}, although it is called a set cover problem in~\cite{Chinneck, Parker}). The MIN IIS COVER set (hitting set) is the smallest set of constraints whose removal makes the original infeasible system feasible. Hence, the MIN IIS COVER problem is identical to the MIN ULR problem, and hence the MAX FS problem.

Parker gives a method for the MAX FS problem in~\cite{Parker}, and Pfetsch further develops this method in~\cite{Pfetsch}. Due to the fact that the infeasible system could contain an exponential number of IISs with respect to the number of constraints and the number of variables~\cite{Chakravarti}, the main idea of Parker's method is finding a subset of IISs in the whole problem and solving an integer program to find a minimum hitting set in each iteration. If the hitting set hits all IISs in the original infeasible system, the optimal solution is found. If not, find some IISs that are not hit by the current hitting set, then find (with an integer program) a new minimum hitting set that also hits the new IISs.

An important part of this method is finding IISs. Given a linear system $Ax \geq b$, where $A \in \mathbb{R}^{m \times n}$ and $b \in \mathbb{R}^{m}$, the following polyhedron:
\begin{equation}
  \label{eq:2.3}
  P = \{y \in \mathbb{R}^{m} | y^{T}A = 0, y^{T}b = 1, y \geq 0\}
\end{equation}
is defined as the \emph{alternative polyhedron}. Each vertex of $P$ corresponds to an IIS in the original infeasible system~\cite{Gleeson,Khach,Parker,Pfetsch}. More precisely, the index set of non-zero supports of a vertex corresponds to an IIS.

\section{Mixed Integer Program (MIP) Formulation}
\label{sec:mip}
Parker suggests two integer program formulations for the MIN IIS COVER problem in~\cite{Parker}. One is applying the big-$M$ method (see~\cite{Parker} and~\cite{Pfetsch}) to the inequalities in the infeasible system, and the other is based on the IIS inequalities. Suppose we have a group of data $\{A_{1}, A_{2}, \ldots, A_{n}\}$ and a point $A_{p}$ in Euclidean space $\mathbb{R}^{d}$, and $x$ is the normal vector of the halfspace that defines the halfspace depth of $A_{p}$. Finding the halfspace depth of $A_{p}$ is equivalent to finding the MIN IIS COVER $\Gamma$ of the following system:
\begin{equation}
  \label{eq:mip.iis1}
  \sum_{i = 1}^{d} (A_{j}^{i} - A_{p}^{i}) x_{i} > 0 \qquad \forall j \in \{1, 2, \ldots, n\}
\end{equation}
The depth of $A_{p}$ is $|\Gamma|$. To simplify~\eqref{eq:mip.iis1} , we call the vector $(A_{j} - A_{p})$ $a_{j}$. Then we rewrite~\eqref{eq:mip.iis1} as
\begin{equation}
  \label{eq:mip.iis2}
  a_{j}x > 0 \qquad \forall j \in \{1, 2, \ldots, n\}
\end{equation}

In~\cite{Parker}, Parker treats the MIN IIS COVER problem with the IIS inequalities. MIN IIS COVER is a minimum hitting set problem, and the hitting set has at least one constraint in common with every IIS in the infeasible system. For an IIS $C$ in \eqref{eq:mip.iis2}, we can use the binary variables associated with the constraints in $C$ to formulate an inequality like
\begin{equation}
  \label{eq:mip.iisinq}
  \sum_{t \in C} s_{t} \geq 1
\end{equation}
where $s_{t}$ is the binary variable associated with constraint $t$ in \eqref{eq:mip.iis2}. Using the IIS inequalities, a hitting set integer program is formulated in the following form:
\begin{eqnarray}
  \label{eq:mip.iismip}
  \textrm{minimize} \qquad \sum_{i = 1}^{n} s_{i} & & \nonumber \\
  \textrm{subject to} \qquad
  \sum_{i \in C} s_{i} & \geq & 1 \qquad \forall C \quad \textrm{(IIS of system~\eqref{eq:mip.iis2})} \\
  s_{i} & \in & \{0 , 1\} \qquad \forall i \in \{1, 2, \ldots, n\} \nonumber
\end{eqnarray}
As we mentioned in Section~\ref{sec:mfs.iis}, Parker's strategy is to first find a small set of IISs and formulate an integer program (a sub-program of \eqref{eq:mip.iismip}). After obtaining the optimal solution to the initial integer program, find some IISs that are not hit by the solution, add the corresponding IIS inequalities into the integer program and resolve it. The process stops when the solution hits all IISs in the infeasible system. This is easy to verify since it means exactly that the remaining system is feasible.

To formulate an integer program with the big-$M$ method, the strict inequalities in \eqref{eq:mip.iis2} need to be transformed into non-strict ones. From \eqref{eq:mip.iis2}, we can derive the following possibly infeasible system:
\begin{equation}
  \label{eq:mip.iiseps}
  a_{j}x \geq \epsilon \qquad \forall j \in \{1, 2, \ldots, n\}
\end{equation}
where $\epsilon$ is a small positive real number. A mixed integer program can be formulated with the big-$M$ method as follows:
\begin{eqnarray}
  \label{eq:mip.mip}
  \textrm{minimize} \qquad \sum_{j = 1}^{n} s_{j} & & \nonumber \\
  \textrm{subject to} \qquad
  a_{j}x + s_{j}M & \geq & \epsilon \qquad \forall j \in \{1, 2, \ldots, n\} \\
  s_{j} & \in & \{0 , 1\} \qquad \forall j \in \{1, 2, \ldots, n\} \nonumber \\
  - \infty \leq & x_{i} & \leq + \infty \qquad \forall i \in \{1, 2, \ldots, d\} \nonumber
\end{eqnarray}
Fixing the binary variable $s_{j}$ to $1$ has the effect of removing constraint $j$ from~\eqref{eq:mip.iis2}. The objective function is to minimize the number of constraints that have to be removed in order to find a feasible subsystem of~\eqref{eq:mip.iis2}. For the general MIN IIS COVER problems, the big-$M$ method may not be practical. As Parker and Pfetsch mentioned, the big-$M$ should be big enough to make the infeasible system feasible, but if it is too big, this will cause numerical problems (see~\cite{Parker} for details).

In this paper we investigate the big-$M$ method for the halfspace depth problem. In this formulation, it is easy to find a value for $M$ to make~\eqref{eq:mip.mip} feasible, but the value of $M$ should be large enough to guarantee an accurate result. For example, if $M$ is assigned to $\epsilon$, \eqref{eq:mip.mip} will be feasible for $x=0$, but the optimal solution will not be the MIN IIS COVER of \eqref{eq:mip.mip} because all the binary variables will be forced to $1$\label{page:bigm}. Now we describe our methods to choose values for $\epsilon$ and $M$ in~\eqref{eq:mip.mip}.

\subsection{Fixing the parameter $M$}
\label{sec:mip.m}

Since we only care about the sign of $a_{j}x$, instead of using infinity, we can choose an arbitrary bound $-c \leq x_{i} \leq c$ for each element of vector $x$, where $c$ is a constant. From the definition of inner product, we can get the following observation:
\begin{equation}
  \label{eq:3.6}
   x \cdot q = \lVert x \rVert \cdot \lVert q \rVert \cdot \cos\alpha \leq \lVert x \rVert \cdot \lVert q \rVert
\end{equation}
Suppose point $q_{max}$ is the vector with the largest norm in the data set, then $M$ can be set to the value of $\sqrt{d \cdot c^{2}} \cdot \lVert q_{max} \rVert$. Actually, we can assign a different value to $M$ in each constraint, namely $\sqrt{d \cdot c^{2}} \cdot \lVert q \rVert$ in the constraint associated with vector $q$. Even further, we can scale our vectors $q$, and make all data have the same norm.

\subsection{Fixing the parameter $\epsilon$}
\label{sec:mip.epsilon}
Now let us consider the value of $\epsilon$. When computing the halfspace depth, maximizing the number of points contained in an open halfspace is the same as maximizing the number of points contained in a cone with solid angle less than $\pi$. For instance, if the open halfspace in Figure~\ref{fig:hdhd} is replaced by a cone with angle close enough to $\pi$ (see Figure~\ref{fig:hdcone}), we can still have the same depth value for $p$.
\begin{figure}[!ht]
  \centering
  \subfigure[]{\label{fig:hdhd}
    \includegraphics[width=0.30\textwidth]{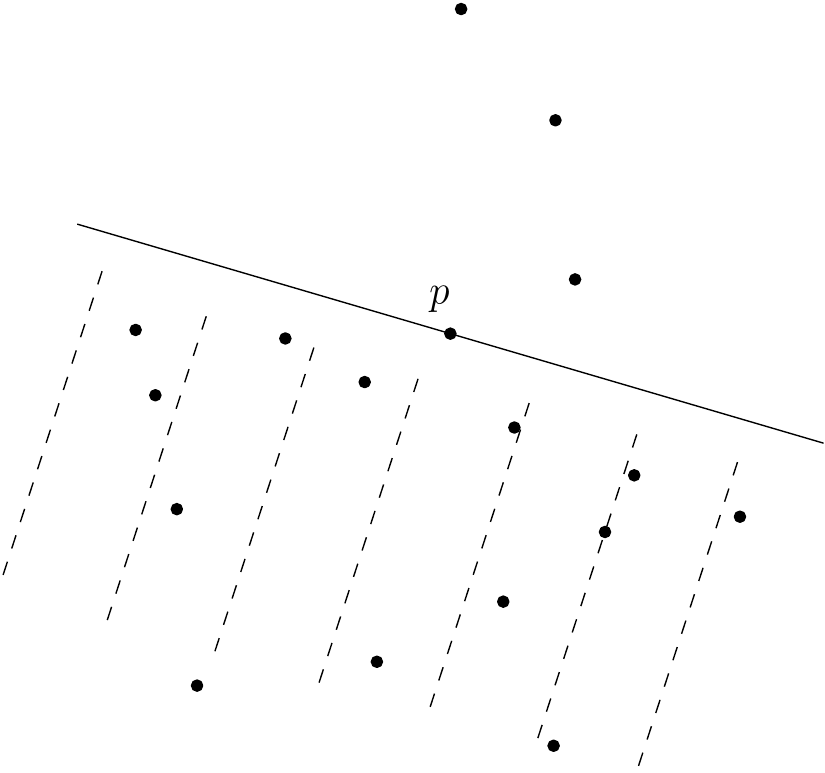}
  }
  \hspace{10mm}
  \subfigure[]{\label{fig:hdcone}
    \includegraphics[width=0.32\textwidth]{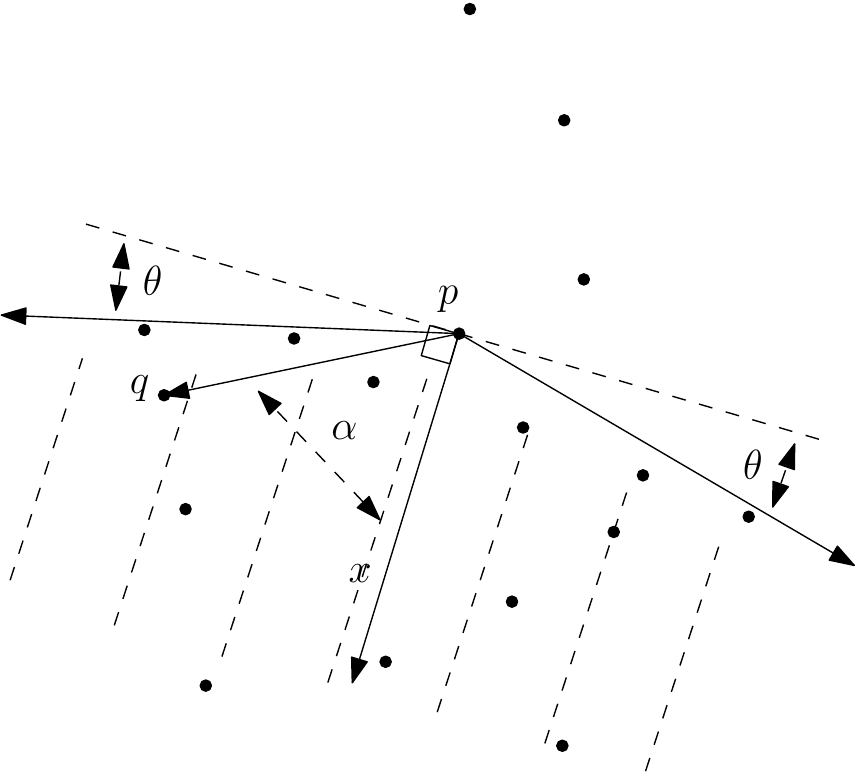}
  }
  \caption{Halfspace depth defined by an open halfspace and a cone}
  \label{fig:3.1}
\end{figure}

Let the angle between the boundary of the cone and the halfspace be $\theta$, $x$ be the outward normal vector of the halfspace, $q$ be a point contained by the cone, and $\alpha$ be the angle between $x$ and $q$ (interpreting $q$ as a vector, and $p$ as the origin). The definition of the inner product tells us
\begin{equation}
  \label{eq:3.4}
  x \cdot q = \lVert x \rVert \cdot \lVert q \rVert \cdot \cos\alpha = \lVert x \rVert \cdot \lVert q \rVert \cdot \sin(\frac{\pi}{2} - \alpha) \geq \lVert x \rVert \cdot \lVert q \rVert \cdot \sin\theta
\end{equation}
where the last inequality is illustrated in Figure~\ref{fig:hdcone}. Suppose point $q_{min}$ is the point with the smallest norm in the data set; then, no matter what value $x$ is in the optimal solution of \eqref{eq:mip.mip}, the following inequality is satisfied for any point $q$ in the data set.
\begin{equation}
  \label{eq:3.5}
  x \cdot q \geq \lVert x \rVert \cdot \lVert q_{min} \rVert \cdot \sin\theta
\end{equation}

Now the question is how to find the $\theta$. Given a value for $\theta$, $\epsilon$ can be set to the value of 
\begin{equation*}
\sqrt{d \cdot c^{2}} \cdot \lVert q_{min} \rVert \cdot \sin\theta\,.
\end{equation*}
The value of $\theta$ can be bounded using the theory of integer lattices. In this argument all the data are considered as integral data (fractional data can be scaled up to integral data), i.e.\ the input data set $S$ will be a subset of the integer lattice~\cite{Gruber}. The following theorem is due to Achill Sch\"urmann~\cite{Schuermann06}.
\begin{theorem}
  Suppose points $\{X_{1}, X_{2}, \ldots, X_{d}\}$ are affinely independent in $\mathbb{Z}^{d}$, $\lvert X_{i}^{j} \rvert \leq m$ , ($i, j = 1,2, \ldots, d$). Let $H$ be an affine hull of $\{X_{1}, X_{2}, \ldots, X_{d}\}$, and $H$ does not contain the origin $O$. Then we can have the following bound for the distance from $O$ to $H$, 
\begin{equation}
  \label{eq:3.dist}
   \dist{(H, O)} \geq (2m\sqrt{d})^{-(d-1)} \, .
\end{equation}
\end{theorem}

\begin{proof}
  Let
  \begin{displaymath}
    l := X_{1} + \mathbb{Z}(X_{2} - X_{1}) + \ldots +  \mathbb{Z}(X_{d} - X_{1})
  \end{displaymath}
be a lattice of $\mathbb{Z}^{d}$ within $H$. Let $l_{0} := H \cap
\mathbb{Z}^{d} $, then we have $l \subseteq l_{0}$. The distance $h$
of $H$ to a parallel plane containing lattice points is 
\begin{displaymath}
  \frac{\det{\mathbb{Z}^{d}}}{\det{l_{0}}} = \frac{1}{\det{l_{0}}} \, .
\end{displaymath}
 Then $\dist{(H, O)} \geq h \geq \frac{1}{\det{l_{0}}}$. Since 
${\det{l}}/{\det{l_{0}}} \in \mathbb{N}$, we have $\det{l} \geq \det{l_{0}}$. Therefore, we have $\dist{(H, O)} \geq \frac{1}{\det{l}}$. Let $C_{m} := \{ x \in \mathbb{R}^{d} : \lvert x_{j} \rvert \leq m\}$. According to Hadamard's inequality, we have $\det{l} \leq \prod_{i=2}^{d} \lVert X_{i} - X_{1} \rVert$. Since $\lVert X_{i} - X_{1} \rVert \leq \textrm{diameter}(C_{m}) = 2m\sqrt{d}$, $\dist{(H, O)} \geq (2m\sqrt{d})^{-(d-1)}$.
\end{proof}

Let us now return back to the idea of halfspace depth defined by a cone (see~\eqref{eq:mip.mip}). As shown in Figure~\ref{fig:lattice}, a distance $h$ defines a cone, and $p$ corresponds to the origin $O$ in the former paragraph. The value of $\sin\theta$ will be  at least ${h}/{\textrm{radius}(C)}$.
\begin{figure}[!ht]
  \centering
  \includegraphics[width=0.4\textwidth]{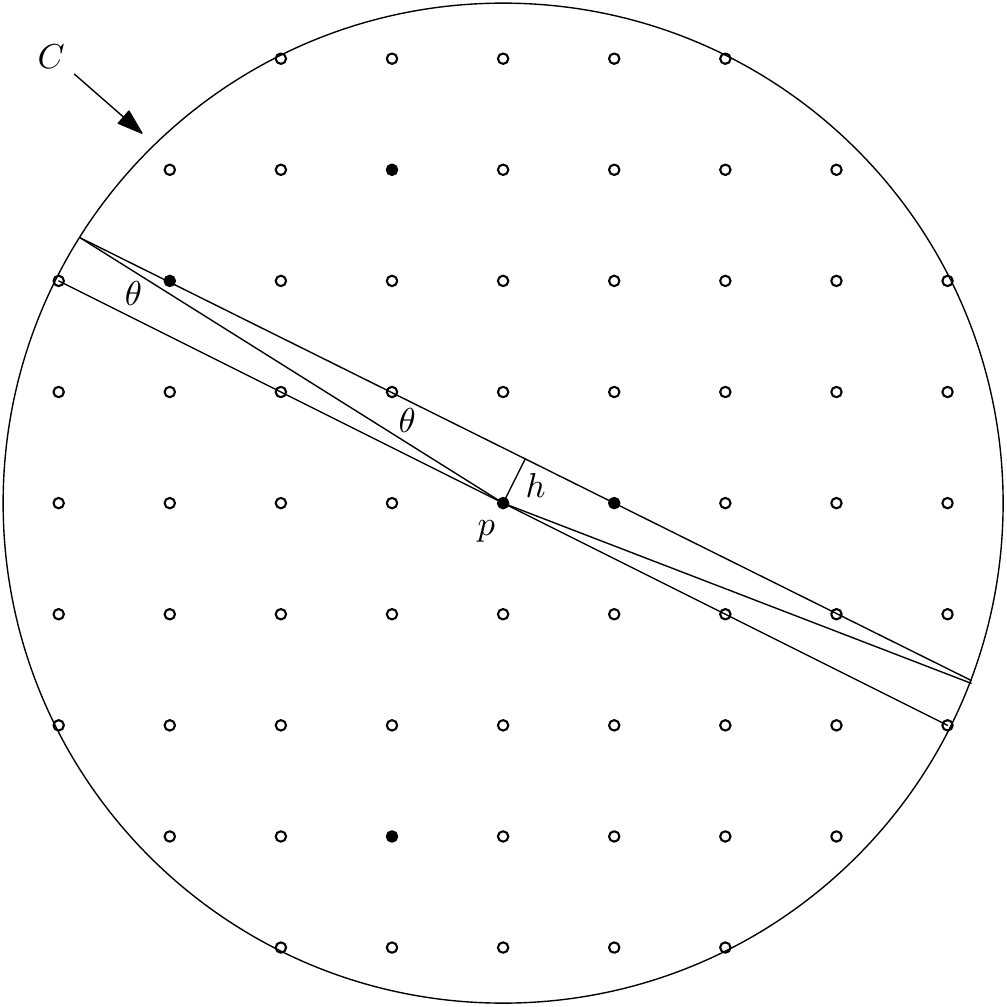}
  \caption{Lattice}
  \label{fig:lattice}
\end{figure}

When the dimension is high, such as $20$, the value of $\epsilon$ based on this lattice idea would be too small to be useful in practice. In our testing, we just set $\epsilon$ to a very small value. If $\epsilon$ is not small enough, it will have the same effect as $M$ not being big enough (see page~\pageref{page:bigm}).

\section{A Heuristic Algorithm}
\label{sec:heur}
Chinneck~\cite{Chinneck2, Chinneck} suggests a heuristic algorithm for the MIN IIS COVER problem. As discussed in Section~\ref{sec:mfs}, it is also an algorithm for the halfspace depth problem. This algorithm is based on several observations of elastic programming (originally a method for solving integer programs~\cite{Brown} according to Chinneck). In elastic programming, every constraint is elasticized by adding a non-negative elastic variable. Chinneck gives the following rules: the constraints in the form of $\sum_{j}a_{ij}x_{j} \geq b_{i}$, $\sum_{j}a_{ij}x_{j} \leq b_{i}$, or $\sum_{j}a_{ij}x_{j} = b_{i}$, become $\sum_{j}a_{ij}x_{j} + e_{i} \geq b_{i}$, $\sum_{j}a_{ij}x_{j} - e_{i} \leq b_{i}$, or $\sum_{j}a_{ij}x_{j} + e_{i}' - e_{i}'' = b_{i}$ respectively. The elastic objective function is to minimize the sum of the elastic variables, which is similar to phase~1 of the \emph{two phase} simplex method~\cite{Chvatal}. After elasticizing, the original infeasible system becomes feasible, and the optimal solution will give some information about the infeasibility in the original system. This elastic programming is also similar to the big-$M$ method. In the big-$M$ method, a set of binary variables with a large coefficient are used to make the infeasible system feasible. When the optimal solution of the elastic program is found, the optimal value of the objective function is called the \emph{sum of the infeasibility (SINF)}. A nonzero elastic variable indicates a violated constraint in the original model, and the number of the nonzero variables is called the \emph{number of infeasibility (NINF)}. As Chinneck observed, the MIN IIS COVER problem is the problem to minimize NINF. At the optimal point, the value of an elastic variable is called the \emph{constraint violation} of the corresponding constraint in the original model. The \emph{reduced cost} of the slack or surplus variable is called the \emph{constraint sensitivity} of the corresponding constraint, which, in fact, is the \emph{shadow price} of that constraint. The shadow price of a constraint indicates how much the objective value of the optimal solution will be changed by changing the right hand side of the constraint by one unit. For more details about elastic programming, please refer to~\cite{Chinneck1}.

The most important observation given by Chinneck is that SINF will be reduced more by eliminating a constraint in the MIN IIS COVER. For a violated constraint in the original model, the drop of the SINF can be estimated by (constraint violation) $\times$ $|$(constraint sensitivity)$|$. For an unviolated constraint, the drop can be estimated by $|$(constraint sensitivity)$|$. Detailed explanations of the observations are available in~\cite{Chinneck2,Chinneck1}. In the heuristic algorithm we basically estimate removing which constraint will reduce the SINF most, then remove that constraint from the infeasible system. We keep repeating this step until the system becomes feasible. The set of removed constraints is an IIS cover set.

\section{Techniques for This Algorithm}
\label{sec:alg}

In our branch and cut algorithm, we first use Chinneck's heuristic algorithm to find a feasible solution and set up the upper bound with this solution. Chinneck's heuristic algorithm is very fast and accurate. Most of the time this heuristic finds an optimal solution. Hence, we will have a very good upper bound at the beginning. The heuristic in~\cite{Chinneck} is used, because it is faster according to Chinneck. We apply IIS hitting set inequalities as cutting planes for the problem. The cuts are problem-specific. 

\subsubsection*{Basic infeasible subsystem cuts}
The paper \cite{Bremner1} describes how to generate a \emph{basic infeasible subsystem (BIS)} of an infeasible system. Given an infeasible system $Ax \geq b$ where $A \in \mathbb{R}^{n \times d}$ and $b \in \mathbb{R}^{n}$, the basic infeasible subsystem is an infeasible subsystem of cardinality no more than $d + 1$. To find a basic infeasible subsystem, the idea is to apply phase~1 of the two phase simplex method~\cite{Chvatal} by solving the following LP:
\begin{eqnarray}
  \label{eq:impl.phase1}
  \textrm{minimize} \qquad  x_{0} \quad \quad & & \nonumber \\
  \textrm{subject to} \qquad
  Ax +x_{0} & \geq & b
\end{eqnarray}
After getting the optimal solution, the set of tight constraints corresponds to a basic infeasible subsystem of $Ax \geq b$. For more details about basic infeasible subsystems, please refer to~\cite{Bremner1}. A basic infeasible subsystem may not be irreducible if it contains a degenerate IIS whose cardinality is smaller than $d + 1$.  On the other hand, since every IIS is basic, it suffices to find a hitting set for the BISs.

\subsubsection*{Pseudo-Knapsack Technique for Generating Cuts}
In order to generate cuts that are violated by the current solution of the LP relaxation, we use a pseudo-knapsack technique to find as many binary variables as possible with a summation smaller than $1$ (note that the binary variables will become continuous variables in the LP relaxation, and with bounds $0 \leq x_{i} \leq 1$ for any variable $x_{i}$). After solving a LP relaxation, the binary variables are ranked according their values in increasing order. Select the first $k$ variables ($k$ is maximal) such that the summation of them is smaller than $1$. Find the IISs in the corresponding constraints (in \eqref{eq:mip.iiseps}) of these variables. Such an IIS must give a violated cutting plane for the current solution of the LP relaxation.

In fact, identifying the maximum set of binary variables is not a true knapsack problem, because in this problem the cost and the value of an item (a binary variable) are the same. The greedy method in the above paragraph will give the optimal solution of this pseudo-knapsack problem. We can prove this by contradiction. Suppose $\{a_{1}, a_{2}, \ldots, a_{n}\}$ is the set of the values of the binary variables in increasing order, the greedy method identifies the first $k$ items, and a better algorithm identifies a set $J$ of $j$ items ($j > k$). The sum of any $k + 1$ items in $J$ is greater or equal to $\sum_{i=1}^{k+1}a_{i}$, because if $J$ contains any items that are different from the items in $\{a_{1}, a_{2}, \ldots, a_{k + 1}\}$, any of those different items would be greater or equal to $a_{k + 1}$. Hence, the sum of the items in $J$ would be greater than $1$, noting that $\sum_{i=1}^{k+1}a_{i} > 1$. Therefore, a better algorithm cannot exist.

This technique is used in one of the two hitting set cut generators we implemented.

\subsubsection*{Branching Variable Selecting Rule}
When selecting the branching variable for a subproblem of \eqref{eq:mip.mip}, we mimic the technique in Chinneck's heuristic algorithm. Let $S_{1}$ be the set of constraints of \eqref{eq:mip.iiseps} that correspond to the constraints of the subproblem. After solving an elastic program of $S_{1}$, we estimate the drop of SINF that each constraint can give (see Section~\ref{sec:heur}). The constraint $b$ which can give the most significant drop has the best chance to be a member of the MIN IIS COVER of \eqref{eq:mip.iiseps} according to Observation~3 in~\cite{Chinneck}. The binary variable $s_{b}$ that corresponds to $b$ is selected as the branching variable.

\subsubsection*{Candidate Problem Selection}
By fixing $s_{b}$, we get two new candidate problems, one with $s_{b} = 1$, the other with $s_{b} = 0$. In the problem selection step, a depth first strategy is used and the problem with $s_{b} = 1$ is selected as the new problem to process. As mentioned in Section~\ref{sec:mip}, fixing the binary variable $s_{b}$ to $1$ has the effect of removing constraint $b$ from \eqref{eq:mip.iis2}. As the algorithm continues to dive in the problem tree, \eqref{eq:mip.iis2} will usually become feasible quickly due to the accuracy of Chinneck's algorithm. At that point, the candidate problem will be fathomed because the optimal objective value will be $0$, an integral solution. This strategy will hopefully keep the depth of the search tree small, so then we would have a good chance to have small search tree for the whole problem.

\subsection{A Binary Search Strategy}
\label{sec:alg.bin}

The idea of fixing the $\epsilon$ in Section~\ref{sec:mip} cannot be used in practice, and we can not guarantee the accuracy of the solutions by assigning the $\epsilon$ an arbitrary small number. However, we can find an accurate solution for a problem by solving several MIPs. In this idea, the MIP~\eqref{eq:mip.mip} needs to be changed to the following form:
\begin{eqnarray}
  \label{eq:alg.bin}
  \textrm{minimize} \qquad -\epsilon & & \nonumber \\
  \textrm{subject to} \qquad \qquad \qquad \qquad
  \sum_{j = 1}^{n} s_{j} & \leq & \mathrm{guess} \nonumber \\
  a_{j}x + s_{j}M & \geq & \epsilon  \qquad \forall j \in \{1, 2, \ldots, n\}\\
  \epsilon & \geq & 0 \nonumber \\
  s_{j} & \in & \{0 , 1\} \qquad \forall j \in \{1, 2, \ldots, n\} \nonumber \\
  - c \leq & x_{i} & \leq c \qquad \forall i \in \{1, 2, \ldots, d\} \nonumber
\end{eqnarray}
In this formulation $\epsilon$ is a variable, and there is also one more constraint in which $\mathrm{guess}$ is a value we want to test the depth against. If the optimal value of the objective function is $0$, $\mathrm{guess}$ is smaller than the depth of point~$A_{p}$. Therefore, we can use binary search to try different values of $\mathrm{guess}$.

In the binary search algorithm our branch and cut algorithm is used as a subroutine which solves (partially) a MIP per invocation. The subroutine will finish as soon as it finds a feasible solution which gives a nonzero $\epsilon$, because a nonzero $\epsilon$ implies that $\mathrm{guess}$ is no less than the depth of $A_{p}$. The binary search algorithm maintains a cut pool containing the cutting planes generated in the early subroutines. The cuts will be used as indexed cuts for later subroutines.

\section{Implementation}
\label{sec:impl}
Our algorithm is implemented with the BCP library from the \textbf{COIN-OR} project~\cite{Coin}, along with the Osi, Clp and Cgl libraries from this project. For the binary search algorithm, we just make some adjustments to the branch and cut algorithm, and use it as a subroutine. BCP is a set of C++ classes and functions which manage the search tree. It does not contain any LP solver or cutting plane generator. The Osi (Open Solver Interface) library is used as the interface between BCP and an LP solver. Clp (COIN-OR linear programming) is used as the LP solver in our implementation. Some commercial LP solvers, like CPLEX or Xpress, might be faster than Clp, but we want other researchers to have easy access to our codes. Of course, it is possible to change the code in order to use other LP solvers thanks to Osi. Cgl (Cut Generation Library) is a collection of cut generators, which is used to generate cutting planes for BCP. BCP is designed for parallel execution in the master slave paradigm; it also supports sequential execution.

Our implementation is based on the example \textbf{BAC}~\cite{Margot} written by Margot. We also implemented Chinneck's heuristic algorithm~\cite{Chinneck} and two cut generators. The cut generators will receive the solutions of an LP relaxation from the LP process and generate cutting planes based on these solutions.

Bremner, Fukuda, and Rosta developed a primal-dual algorithm for the halfspace depth problem in~\cite{David}. Their algorithm is to find the minimum traversal of all MDSs in the input data set. They developed a library to generate the MDSs (recall that MDSs are [the same as] IISs), and that library is based on Avis' \textbf{Lrslib} library~\cite{Avis}. We use this MDS generating library to generate IISs for our algorithm in the first cut generator. In this cut generator, we use the pseudo-knapsack technique in Section~\ref{sec:alg} to find a set of binary variables. Then we identify the set of points in the input data set that correspond to the binary variables. This set of points are then used as the input for the MDS generate library to generate a set of IISs. Finally we formulate a set of cutting planes in the form of \eqref{eq:mip.iisinq}, one for each IIS found.

BCP does not support global cuts currently. Any cuts added to a subproblem are only available to its children. This is unfortunate for us, since all hitting set cuts are globally valid. On the other hand, keeping too many cuts can slow down each node (Bremner, Fukuda, and Rosta~\cite{David} observed adding all IIS cuts made solving the LP relaxation very slow). To use the candidate problem selection strategy in Section~\ref{sec:alg}, we set the tree search strategy to the depth first search and the child preference to dive down.

\section{Computational Experiments}
\label{sec:test}
Our algorithms have been tested on a Myrinet/4-way cluster that consists of dual socket SunFire x4100 nodes which are populated with 2.6 GHz dual-core Opteron 285 SE processors and 4 GB RAM per core. We set the CPU time limit to 60 minutes in these tests. For readability, we can not list the raw experimental results, and report only a summary in this section.

In practice, if the value of the $\epsilon$ in the MIP is too small compared with the coefficients of the constraints, the linear programming solver would round it to zero. Our remedy is scaling the data items, and making the norms similar and relatively small. For a few data sets, the depth values reported by our algorithm with different strategies or parameters are different (with a difference of $1$). This could be caused by bugs in our codes or bugs in BCP, but we also suspect this is due to some numerical issue.

\subsection{Results for Random Generated Data}
\label{sec:test.rand}
The data sets tested in this section are a subset of the data sets used in~\cite{David}, and they are randomly generated. For every data set we compute the depth of the first point, which is the origin. For all the tests in this section, the $\epsilon$ of the MIP~\eqref{eq:mip.mip} is set to $0.00001$. Comparing with the results of the primal-dual algorithm and the binary search algorithm, the depth values computed with our branch and cut algorithm (with BIS cut generator) are accurate. Therefore the $\epsilon$ is small enough.

\subsubsection{Comparing Branching Rules}
\label{sec:test.rand.bran}
We first test our algorithm with the first hitting set cut generator in Section~\ref{sec:impl}, the one implemented with the MDS generating library, and with the greedy branching rules (see Section~\ref{sec:alg}). We generate $10$ cuts in one iteration of the LP process. If the objective value is improved by $0.001$, the LP process will do another iteration. When the MDS cut generator is used, most of the CPU time is spent on cutting plane generation. If more cuts are generated in one iteration, the algorithm will be slowed down, but will be more memory efficient. In Figure~\ref{fig:test.bran-n50} and Figure~\ref{fig:test.bran-sd5} we compare the performance of BCP's default strong with the greedy branching rules. The figures show (particularly in Figure~\ref{fig:test.bran-sd5}) that strong branching gives a little better performance for most problems, probably because less search tree nodes are processed. For many difficult problems, the greedy branching works better. In those difficult cases, greedy branching spent much less time on branching, although more search tree nodes would be processed.
\begin{figure}[!ht]
  \centering
  \includegraphics[width=0.6\textwidth]{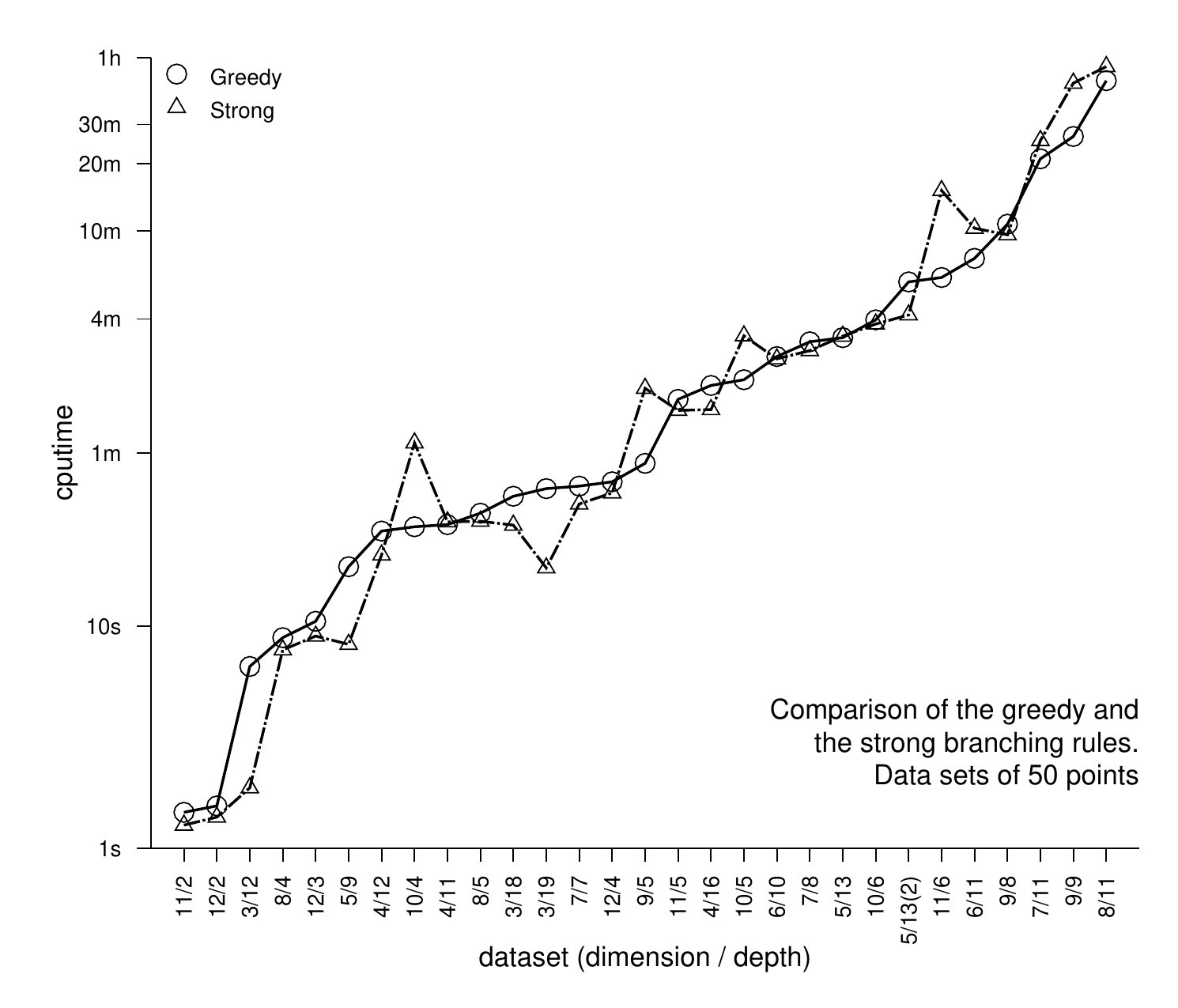}
  \caption{Comparison of different branching rules}
  \label{fig:test.bran-n50}
\end{figure}

\begin{figure}[!ht]
  \centering
  \includegraphics[width=0.6\textwidth]{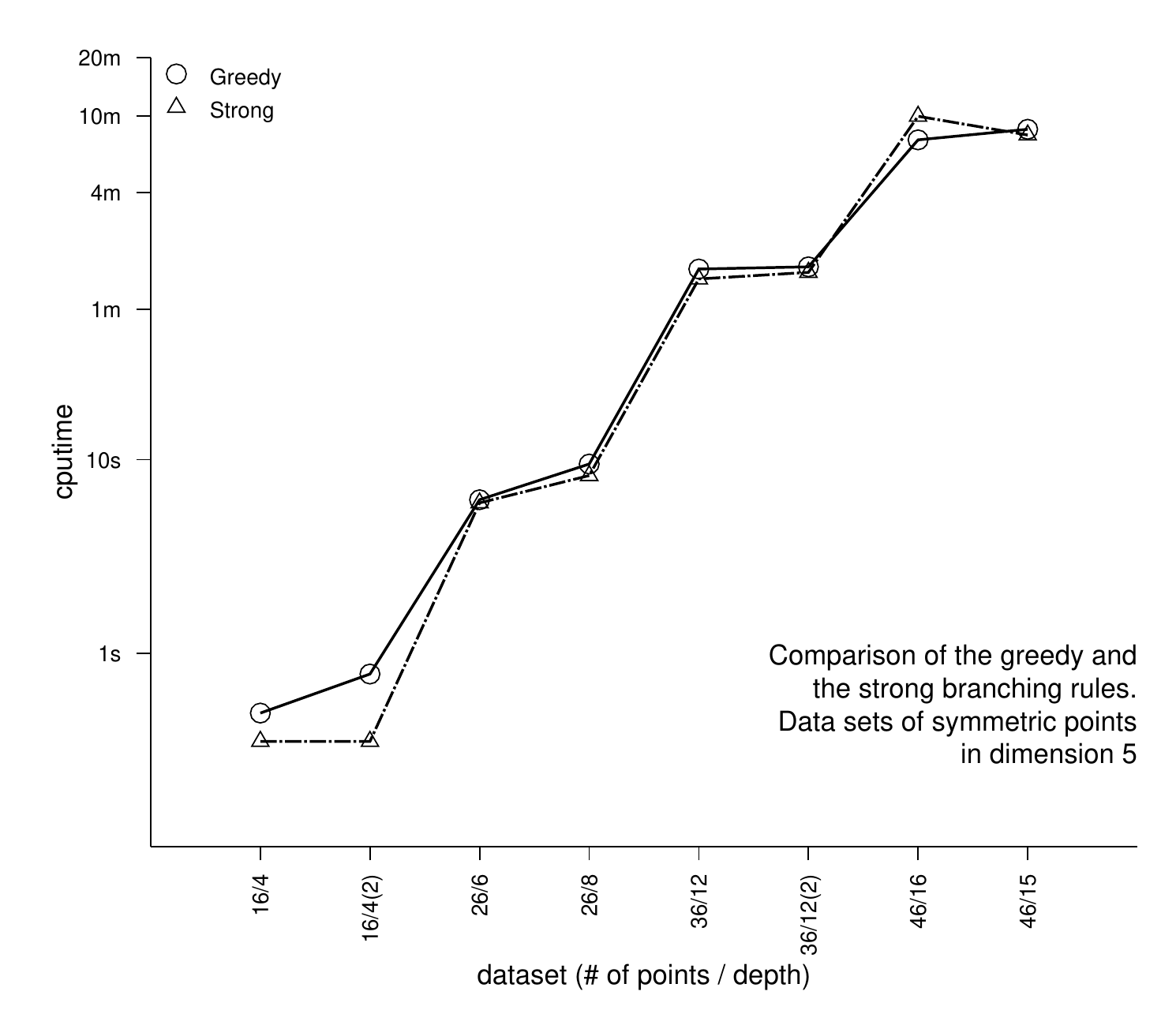}
  \caption{Comparison of different branching rules}
  \label{fig:test.bran-sd5}
\end{figure}

\subsubsection{Comparing Cut Generators}
\label{sec:test.rand.cut}

With the BIS cut generator, much less CPU time will be used to generate cuts, and the algorithm has better overall performance, although the search tree is larger. The BIS cut generator uses floating point arithmetic, the same as the rest of the system. The MDS cut generator uses exact arithmetic which is required for \textbf{Lrslib}. This is a factor which slows down the MDS cut generator. In fact many cuts are generated repeatedly in the optimization process when the BIS cut generator is used. The pseudo-knapsack idea in Section~\ref{sec:alg} can force the algorithm to generate a different cut each time, but the performance turns out to be worse, and more search tree nodes will be processed. With the pseudo-knapsack idea, if the algorithm generates a cut with a probability less than $1$ on each node, the performance will be improved to some extent, although still worse than that without pseudo-knapsack. We also observe that the pseudo-knapsack idea can make the algorithm faster when the greedy branching is applied. This suggests that the pseudo-knapsack idea interferes with the strong branching rule. The reason might be that this idea makes the values of the binary variables in the solution of LP relaxation closer to each other.

\begin{figure}[!ht]
  \centering
  \includegraphics[width=0.6\textwidth]{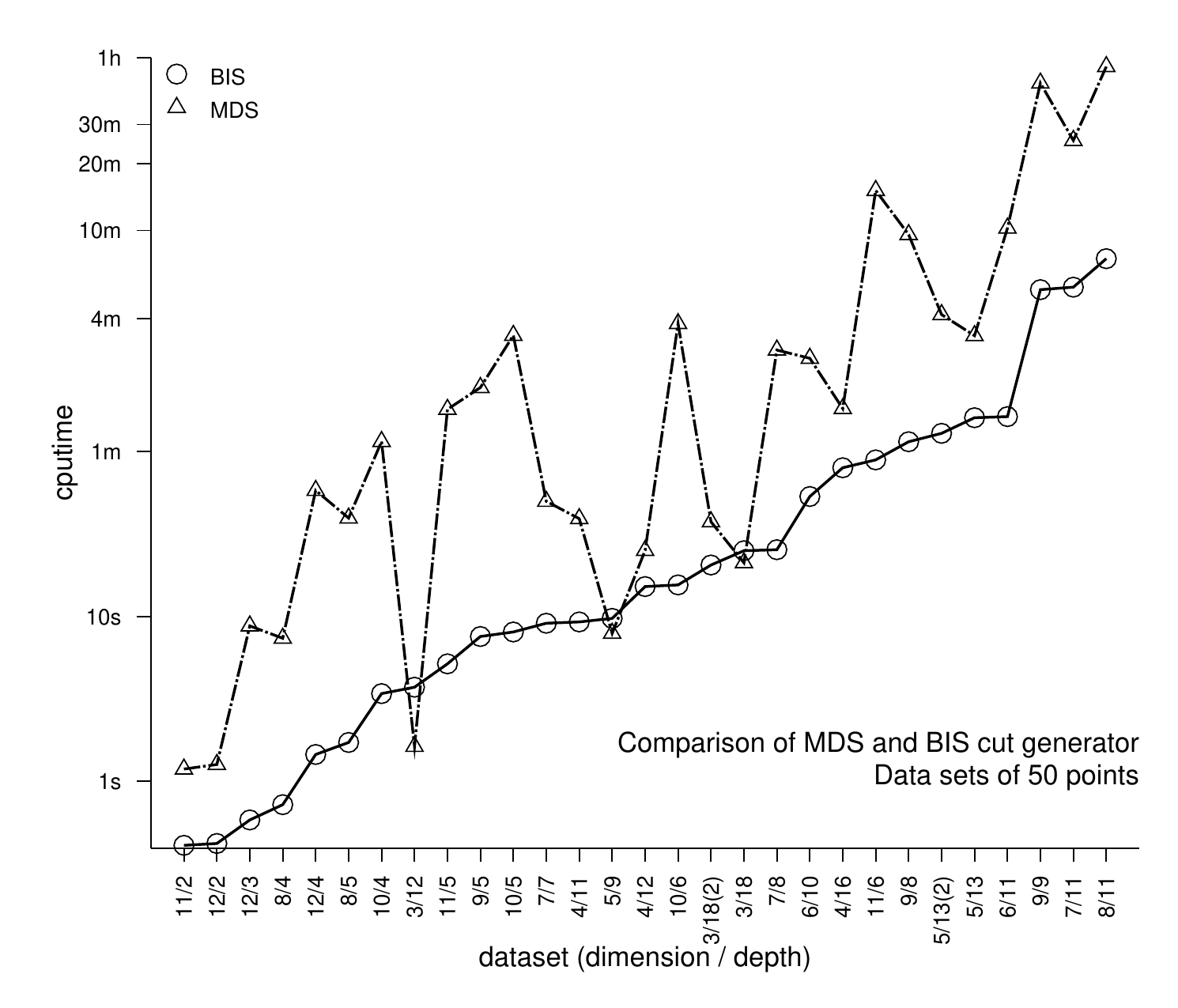}
  \caption{Comparison of different cutting plane generators}
  \label{fig:test.cut-n50}
\end{figure}

\begin{figure}[!ht]
  \centering
  \includegraphics[width=0.6\textwidth]{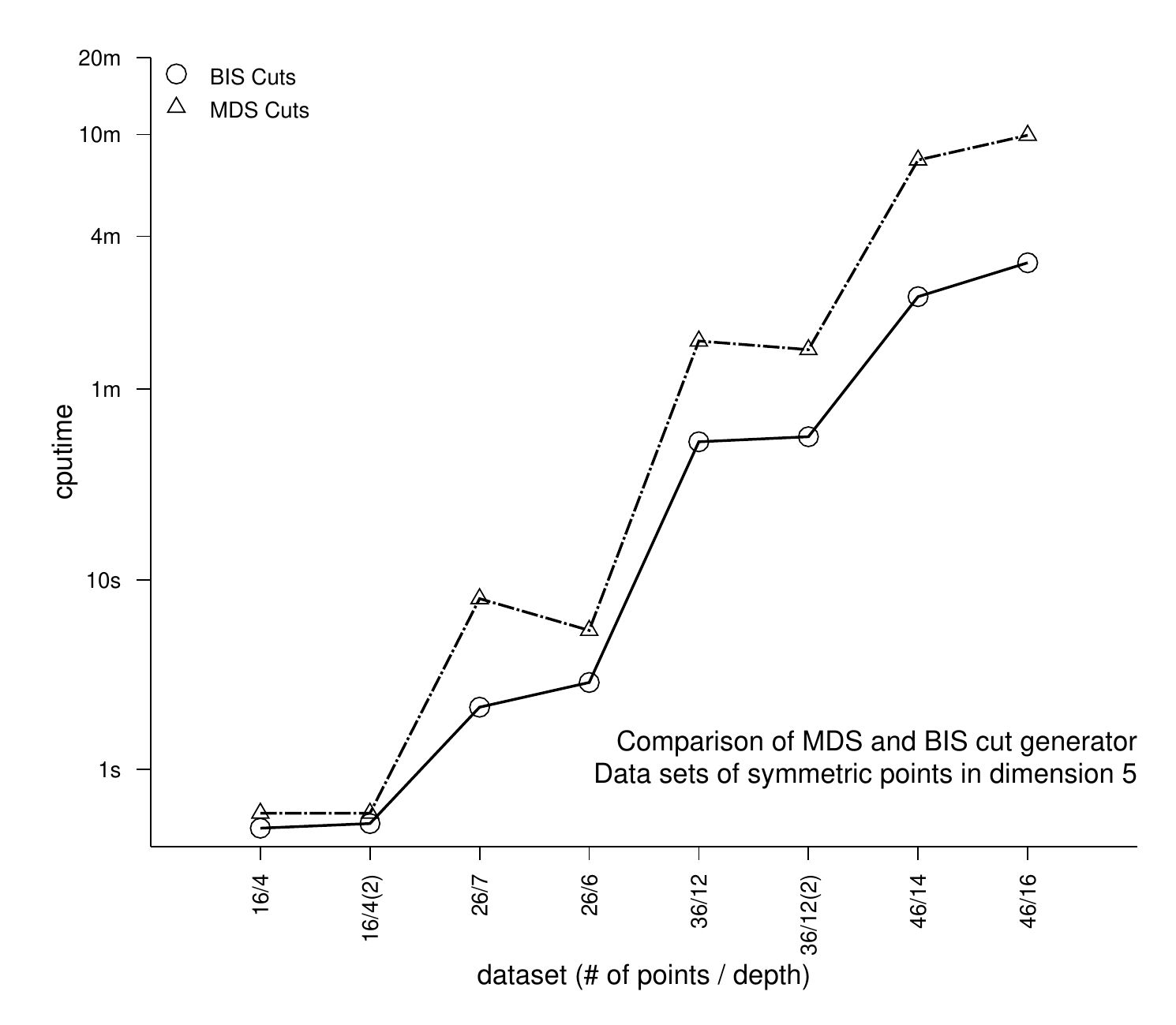}
  \caption{Comparison of different cutting plane generators}
  \label{fig:test.cut-sd5}
\end{figure}

In Figure~\ref{fig:test.cut-n50} and Figure~\ref{fig:test.cut-sd5} we compare the performance of our algorithm with the two different cutting plane generators. The general cut generators in Cgl can barely generate cuts for our algorithm, and do not improve the performance.

\subsubsection{Comparing Algorithms}
\label{sec:test.rand.alg}

The binary search algorithm does not perform too badly, but the primal-dual algorithm is very slow on some hard problems. In Figure~\ref{fig:test.bbp-n50} and Figure~\ref{fig:test.bbp-sd5} we compare the performance of the binary search algorithm, the primal-dual algorithm, and the branch and cut algorithm. The branch and cut algorithm works best most of the time. The performance of the binary search algorithm is actually quite fast (as well as being more numerically stable). Sometimes the binary search algorithm even works faster than the branch and cut algorithm. The reason is that the MIPs for the binary search algorithm are usually easier to solve, and the tricks used in the binary search algorithm also help to speed up the algorithm. In contrast, the primal-dual algorithm can be slow on large problems.

\begin{figure}[!ht]
  \centering
  \includegraphics[width=0.6\textwidth]{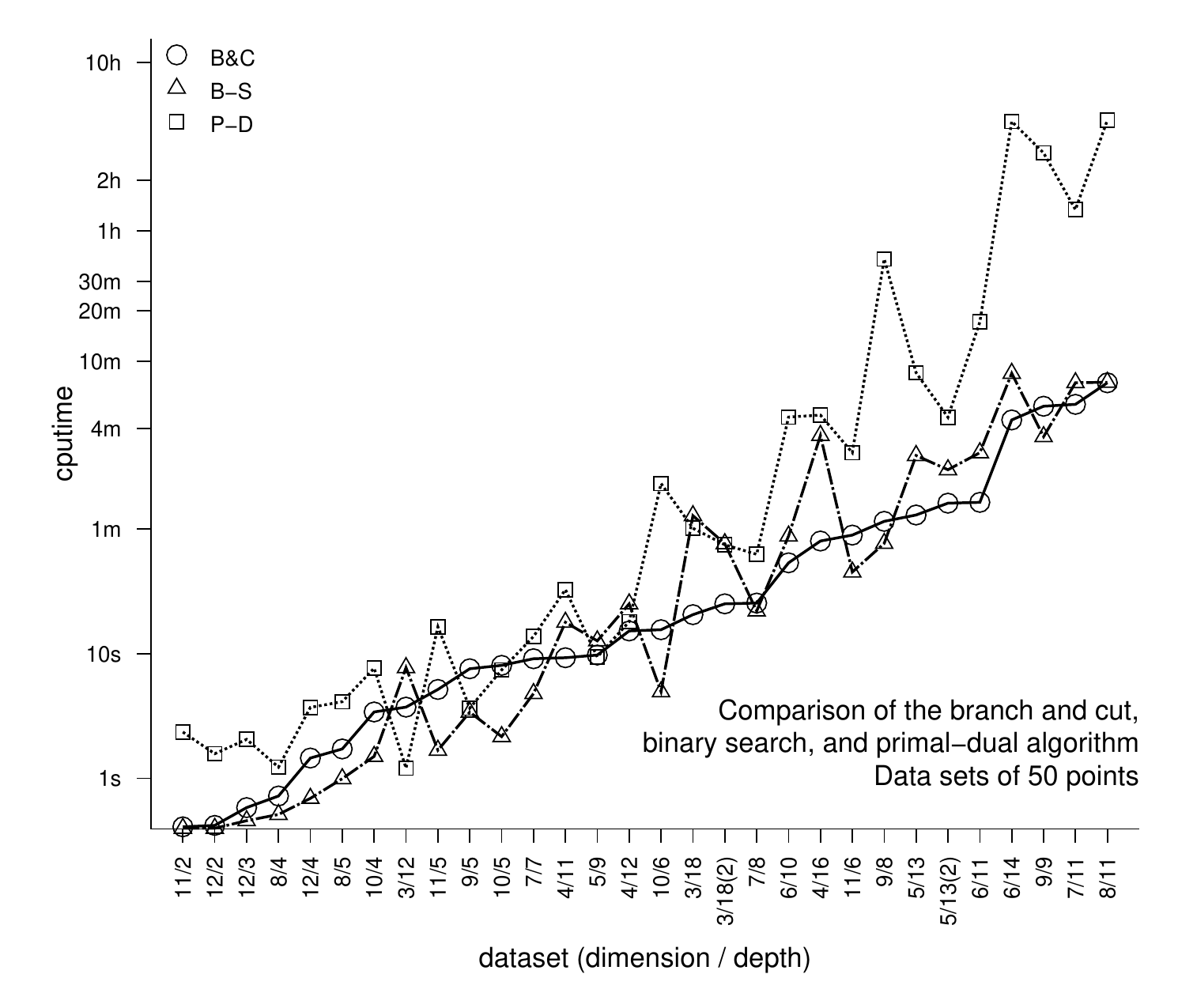}
  \caption{Comparison of different algorithms}
  \label{fig:test.bbp-n50}
\end{figure}

\begin{figure}[!ht]
  \centering
  \includegraphics[width=0.6\textwidth]{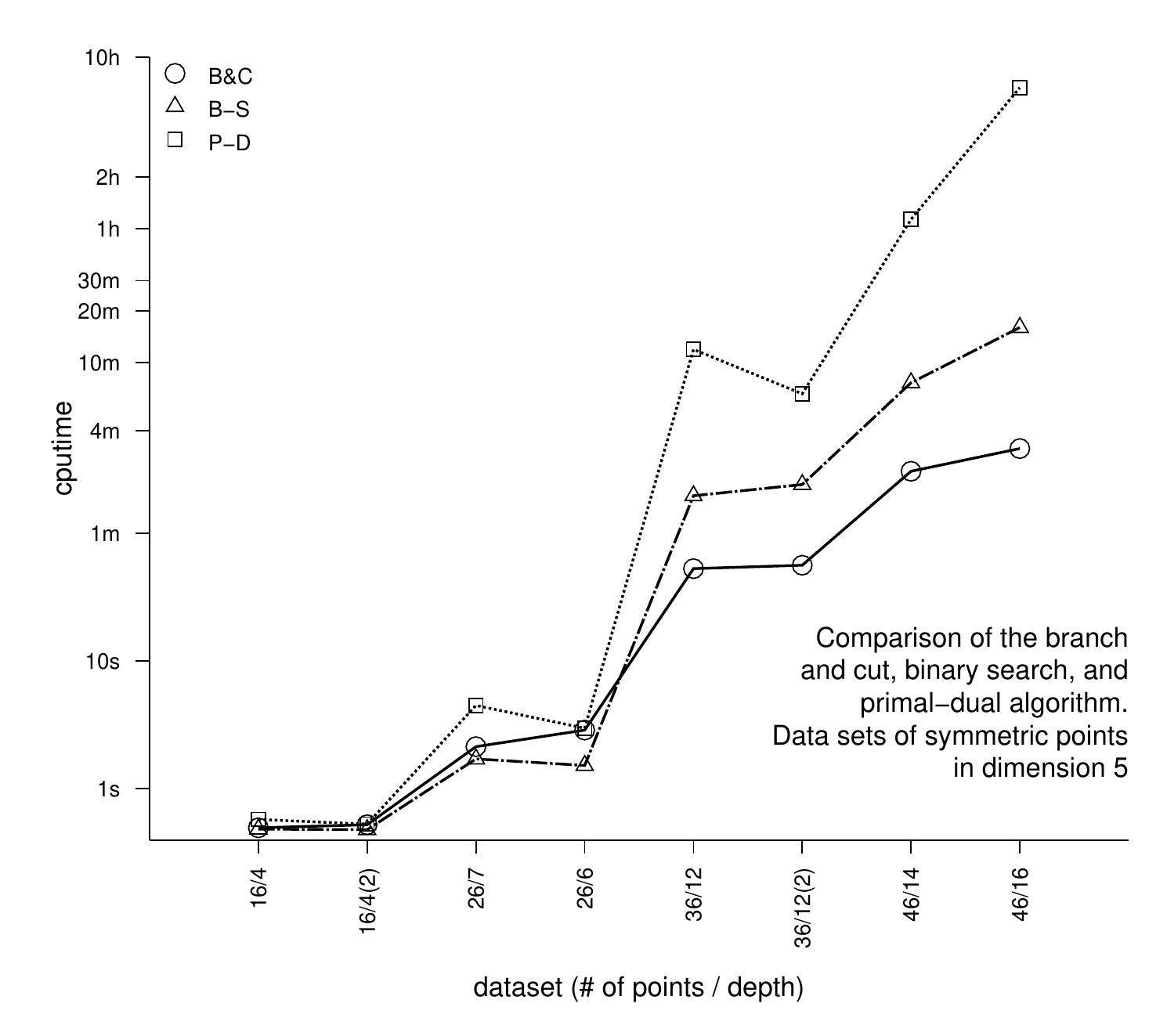}
  \caption{Comparison of different algorithms}
  \label{fig:test.bbp-sd5}
\end{figure}

\subsubsection{Parallel Execution}
\label{sec:test.rand.para}

All the above tests are done with the sequential version of our algorithm. Some tests of parallel version of the branch and cut algorithm are given in Figure~\ref{fig:test.paral}. Two data sets are used to test the algorithm. The performance with one processor is the performance of the sequential version of the algorithm. When two processors are applied, one of them is used for the slave process (LP process), and when four processors are applied, three of them are used for the slave process. So we expect a speedup of $3$ for four processors, $7$ for eight processors, and so forth. The dashed line in the figure indicates the linear speedup with respect to number of LP processes. From the figure we can see that the speedup is almost linear.

\begin{figure}[!ht]
  \centering
  \includegraphics[width=0.6\textwidth]{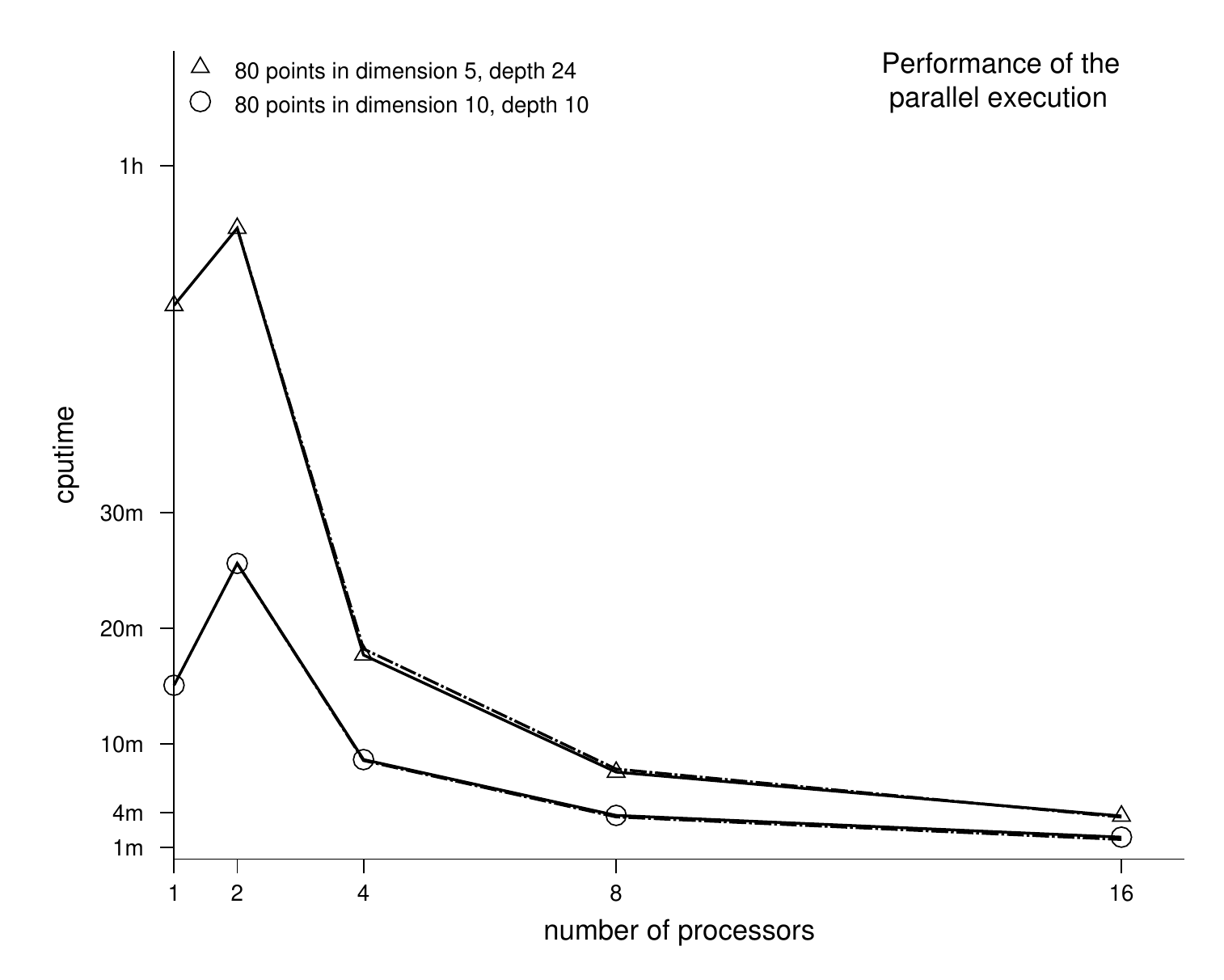}%a better one is required.
  \caption{The Performance of Parallel Execution}
  \label{fig:test.paral}
\end{figure}

\subsection{Results for ANOVA Data}
\label{sec:test.2}
The ANOVA data tested in this section are randomly generated according to the scheme discussed in Section~\ref{sec:anova}. There are some duplicated points in every data set. The same duplicated data are either all inside the halfspace or all outside the halfspace when finding the depth of a point. Therefore, the binary variables associated with these data will be one or zero simultaneously. When formulating the MIP, we keep only one of the duplicated constraints and assign a weight of the number of the duplicated constraints to the associated binary variable in the objective function. For every data set, the depth of the origin is computed. Table~\ref{tab:test.mip} gives a comparison of the performance of our algorithm with different integer program formulations, the simple MIP as~\eqref{eq:mip.mip} and the weighted MIP as described above. The sequential algorithm is used for these tests. The upper bound of the number of duplications (per data point) in the data set is given in the third column. From the table we can see that the algorithm using the weighted MIP is much faster, because there are many fewer rows in the MIP, and correspondingly fewer binary variables.

\begin{table}[!htb]
  \centering
    \begin{tabular}[center]{|c|c|c|c|c|c|}
    \hline
    Point \# & Dimension & Duplication & Depth & Simple MIP & Weighted MIP \\
    \hline
    32 & 8 & 2 & 5 & 0.62 & 0.19 \\%4-4-2.1
    32 & 8 & 2 & 10 & 5.89 & 2.06 \\%4-4-2.2
    32 & 8 & 2 & 7 & 0.90 & 0.24 \\%4-4-2.3
    32 & 8 & 2 & 7 & 0.54 & 0.31 \\%4-4-2.4
    32 & 8 & 2 & 4 & 0.14 & 0.03 \\%4-4-2.5
    48 & 8 & 3 & 5 & 0.11 & 0.08 \\%4-4-3.1
    48 & 8 & 3 & 11 & 4.39 & 0.82 \\%4-4-3.2
    48 & 8 & 3 & 10 & 2.17 & 0.62 \\%4-4-3.3
    48 & 8 & 3 & 9 & 2.04 & 0.28 \\%4-4-3.4
    48 & 8 & 3 & 13 & 27.72 & 2.00 \\%4-4-3.5
    64 & 8 & 4 & 11 & 1.92 & 0.22 \\%4-4-4.1
    64 & 8 & 4 & 17 & 91.89 & 2.98 \\%4-4-4.2
    64 & 8 & 4 & 18 & 200.82 & 3.48 \\%4-4-4.3
    64 & 8 & 4 & 15 & 30.85 & 0.87 \\%4-4-4.4
    64 & 8 & 4 & 16 & 28.94 & 1.60 \\%4-4-4.5
    72 & 12 & 2 & 13 & 147.59 & 22.19 \\%6-6-2.1
    72 & 12 & 2 & 18 & 807.20 & 250.77 \\%6-6-2.2
    72 & 12 & 2 & 14 & 85.94 & 33.52 \\%6-6-2.3
    72 & 12 & 2 & 17 & 529.16 & 69.92 \\%6-6-2.4
    72 & 12 & 2 & 20 & outmem & 469.81 \\%6-6-2.5
    108 & 12 & 3 & 26 & outmem & 519.49 \\%6-6-3.1
    108 & 12 & 3 & 24 & outmem & 264.20 \\%6-6-3.2
    108 & 12 & 3 & 24 & outmem & 341.87 \\%6-6-3.3
    108 & 12 & 3 & 29 & outmem & 1435.35 \\%6-6-3.4
    108 & 12 & 3 & 22 & outmem & 105.49 \\%6-6-3.5
    144 & 12 & 4 & 33 & outmem & 1238.99 \\%6-6-4.1
    144 & 12 & 4 & 39 & outmem & 1760.49 \\%6-6-4.2
    144 & 12 & 4 & 40 & outmem & 1527.83 \\%6-6-4.3
    144 & 12 & 4 & 33 & outmem & 544.95 \\%6-6-4.4
    144 & 12 & 4 & 29 & outtime & 330.57 \\%6-6-4.5
    \hline
  \end{tabular}
  \caption{Performance with different integer program formulations}
  \label{tab:test.mip}
\end{table}

\section{Conclusions and Directions for Future Work}
\label{sec:fuwok}

Comparing the branch and cut algorithm with the binary search algorithm and the primal-dual algorithm, we conclude that the branch and cut algorithm is the fastest, although with some numerical issues. The binary search is slower, but still faster than the primal-dual algorithm and more stable. Fast cutting plane generators are important, because the BIS cut generator improves the performance dramatically. The strong branching rule is a little faster than the greedy branching rule on most of the tests, but the greedy branching rule is faster on many hard problems (i.e.\ those with large depth). The branch and cut algorithm has almost linear speed up for parallel execution. On ANOVA data sets, the duplicated constraints are removed with the weighted MIP formulation. With this modification, the algorithm solved all the problems we tested.

In some applications, only the median of the data set is interesting. With the current algorithm we have to compute the depth of every data item in order to find the median. Finding a fast algorithm for computing the median is open for future work. 

The idea for finding a proper $\epsilon$ described in Section~\ref{sec:mip} is not practical. Another open problem is a method to find a practical $\epsilon$ for MIP~\eqref{eq:mip.mip}. It may be possible to solve an MIP based on the strict inequalities of system~\eqref{eq:mip.iis2}. Then we do not need to consider $\epsilon$. The binary search algorithm does not require a value for $\epsilon$, and it can report a proper value for $\epsilon$. Ironically, this algorithm finds a proper value after solving the halfspace depth problem.

As we noticed in Section~\ref{sec:test.rand.cut}, the pseudo-knapsack idea slows down the strong branching when using the BIS cut generator. An idea for reducing redundant cut generation in the BIS cut generator that does not interfere with the strong branching would be interesting.

\bibliographystyle{plain}
\bibliography{hsdepth}
\end{document}